\documentclass[reqno,12pt]{amsart}

\usepackage{amssymb,amsthm,amsmath}
\usepackage{color}
\usepackage{graphicx}
\usepackage{tikz}
\usepackage{comment}
\usepackage{amsfonts}
\usepackage{enumitem}

\usepackage[colorlinks, 
linkcolor=red,
anchorcolor=blue,
citecolor=blue
]{hyperref}
\title[]{ Quantum dynamical bounds for long-range operators with skew-shift potentials}

\author{Wencai Liu}
\address{
	Department of Mathematics, Texas A\&M University, College Station, TX, 77843, USA.
}
\email{liuwencai1226@gmail.com, wencail@tamu.edu}

\author{Matthew Powell}
\address{
	Department of Mathematics, Georgia Institute of Technology, Atlanta, GA, 30332, USA.
}
\email{powell@math.gatech.edu}

\author{Xueyin Wang}
\address{
	Department of Mathematics, Texas A\&M University, College Station, TX, 77843, USA.
}
\email{xueyin@tamu.edu}

\newcommand{\Z}{\mathbb{Z}}
\newcommand{\T}{\mathbb{T}}
\newcommand{\R}{\mathbb{R}}

\newcommand{\C}{\mathbb{C}}

\theoremstyle{plain}
\newtheorem{theorem}{Theorem}[section]
\newtheorem{corollary}[theorem]{Corollary}
\newtheorem{lemma}[theorem]{Lemma}

\theoremstyle{definition}
\newtheorem{definition}[theorem]{Definition}

\newtheorem{remark}[theorem]{Remark}

\begin{document}

\begin{abstract}
	We employ Weyl's method and Vinogradov's method to analyze skew-shift dynamics on semi-algebraic sets. Consequently, we improve the quantum dynamical upper bounds of Jitomirskaya-Powell, Liu, and Shamis-Sodin for long-range operators with skew-shift potentials.
\end{abstract}

\maketitle	

\section{Introduction}
In this paper, we are interested in studying the quantum dynamics of long-range operators on the integer lattice $\Z.$ That is, we study bounded self-adjoint operators $H$ acting on $\ell^2(\Z)$ in the following way:
\begin{equation*}
	(Hu)_n = \sum_{n' \in \Z} A(n,n')u_{n'} + V(n) u_n
\end{equation*}
where $V(n)$ is a sequence of real numbers and $A(n, n')$ satisfies, for any $n, n' \in \Z,$
\begin{align*}
	| & A(n,n')| \leqslant C \mathrm{e}^{-c|n - n'|}, \\
	  & A(n, n') = \overline{A(n',n)},\\
   &A(n + k, n' + k) = A(n, n'), \ \ \text{ for all }k\in \Z.                
\end{align*}
Operators such as $H$ are commonly associated with the Hamiltonians of quantum particles that evolve according to Schr\"odinger dynamics. For $p > 0$ and $\phi \in \ell^2(\Z),$ one of the primary objects of interest for such operators is the $p$th moment of the position operator, 
given by 
\begin{equation*}
	\langle| X _H|_\phi^p\rangle(T) = \sum_{n \in \mathbb{Z}} |n|^{p} |(\mathrm{e}^{-\mathrm{i}TH}\phi,\delta_n)|^2
\end{equation*}
and its time average, 
\begin{equation*}
	\langle|\tilde X _H|_\phi^p\rangle(T) = \frac 2 T \int_0^\infty \mathrm{e}^{-2t/T} \langle| X _H|_\phi^p\rangle(t) \ \mathrm{d}t.  
\end{equation*}     
These moments relate to the spread of the wavepacket $\mathrm{e}^{-\mathrm{i}tH}\phi,$ which is in turn closely related to the spectral measure $\mu_\phi.$ 
For example, the celebrated RAGE Theorem of Ruelle, Amrein, Georgescu, and Enss says that 
$$\lim_{T \to \infty} \langle|\tilde X _H|_\phi^p\rangle(T) = \infty$$
if the spectral measure $\mu_\phi$ is not pure point. Refinements of this relation have been obtained in various works; we highlight two: one by Last \cite{LastQD1996} and another more recently by Landrigan-Powell \cite{LandriganPowell} both indicate that continuity of $\mu_\phi$ with respect to a (generalized) Hausdorff measure implies a lower bound on 
$\langle|\tilde X _H|_\phi^p\rangle(T),$ 
typically called ``quasi-ballistic transport'' in the literature,
while upper bounds on 
$\langle|\tilde X _H|_\phi^p\rangle(T)$ 
imply that 
$\mu_\phi$ 
must be singular with respect to a particular (generalized) Hausdorff measure. The converse is not true in general: the singularity of 
$\mu_\phi$ 
alone does not imply anything about the behavior of 
$\langle|\tilde X _H|_\phi^p\rangle(T)$ 
(see e.g. \cite{LastQD1996} or \cite{dRJLS1996}).

It is well-known that the motion of the quantum particle cannot be faster than ballistic:
$$\langle| X _H|_\phi^p\rangle(T) \leqslant C^p T^p.$$
More refined (upper or lower) bounds typically have to be obtained on a case-by-case basis (see, e.g. \cite{DamanikTcherem1, DamanikTcherem2, GuarneriSchulzBaldez, Haeming2024, HanJitomirskaya, JPo, JitomirskayaZhang, LastQD1996, LiuQD23, ShamisSodinQD}, see also the survey \cite{DamanikMalinovitchYoung} and references therein). 
Two fruitful methods, broadly speaking, for obtaining lower bounds include:  (1) a careful study of the spectral measures and/or solutions to an eigenvalue equation (see, e.g. \cite{GuarneriSchulzBaldez, JitomirskayaZhang, LastQD1996} see also \cite{DamanikTcherem1} and the references therein) and (2) approximation via operators with ballistic transport (c.f. \cite{Haeming2024, LastQD1996}). 
On the other hand, upper bounds have been obtained using various methods often inspired by localization proofs \cite{HanJitomirskaya, LiuQD23, JPo, ShamisSodinQD} or involving complex analytic methods \cite{DamanikTcherem1, DamanikTcherem2}. The focus of this paper will be to obtain upper bounds for a large class of operators.

Let us consider a particular family of operators $\{H_x\}_{x\in\Omega}$ with a dynamically-defined potential. Let
$(\Omega, f, \mu)$ 
be a measure-preserving dynamical system, and suppose 
$V(n) = \lambda v(f^nx)$ 
for some $\mu$-measurable function $v: \Omega \to \R$ and some $\lambda \ne 0,$ which we call the coupling constant. Such families are particularly well-studied: 
\begin{enumerate}
	\item When $\Omega$ is a sequence space and $f$ is a Bernoulli shift, we obtain a long-range Anderson model;
	      	      
	\item When 
	      $\Omega = \T^b$, $(1,\omega)$ rationally dependent, and $f$ is the shift by $\omega$, we obtain the periodic operators;
	      	              
	\item When $\Omega = \T^b$, $(1,\omega)$ rationally independent, and $f$ is the shift by $\omega$, we obtain the quasi-periodic operators;
	      	              
	\item When 
	      $\Omega = \T^b$, $\omega \in \mathbb{T}\backslash \mathbb{Q}$, and $f$ is the skew-shift:
	      \begin{equation*}
	      	fx = (x_1 + \omega, x_2 + x_1, x_3 + x_2, \cdots, x_b + x_{b - 1}),
	      \end{equation*}
	      we obtain a model closely related to the kicked-rotor problem.
\end{enumerate}

When investigating quasi-periodic and skew-shift models, it is frequently necessary to impose an arithmetic constraint on the admissible frequencies $\omega$, since the spectral properties of these models depend sensitively on the arithmetic properties of $\omega.$ One such constraint we will employ here is the Diophantine condition $DC(\gamma,\tau)$ with $\gamma>0, \tau>1$. Specifically, we say $\omega \in DC(\gamma,\tau)$ if
\begin{equation*}
	\| k \omega\|_{\mathbb{T}} \geqslant \frac{\gamma}{|k|^\tau}, \ \text{for any}\ k\in \mathbb{Z}\backslash\{0\}.
\end{equation*}

For Anderson models with a large coupling constant, or at a spectral edge, dynamical localization holds ($\langle| X _H|_\phi^p\rangle(T)$ is bounded uniformly in $T$). For periodic models, the spectral measures are absolutely continuous, leading to $\langle| X_{H}|_\phi^p\rangle(T)$ being arbitrarily close to ballistic behavior. As stated above, however, quasi-periodic models with $v$ being real analytic, the dynamics exhibit a delicate dependence on the arithmetic properties of both $\omega$ and $x$. Specifically, when $\omega \in DC(\gamma,\tau)$ and $\lambda$ is sufficiently large (depending only on $v$), dynamical localization is known to hold for a.e. (but not every) $x$. Despite the intricate dynamics, bounds on $\langle| X_{H_{x,\omega}}|_\phi^p\rangle(T)$ exhibit greater stability. It is known that under the conditions $\omega \in DC(\gamma,\tau)$ and $\lambda$ sufficiently large (depending only on $v$),  the inequality 
\begin{equation*}
	\langle| X_{H_{x,\omega}}|_\phi^p\rangle(T) \leqslant C (\log T)^{p\sigma + \varepsilon}
\end{equation*}
holds uniformly in $x$, where $\sigma = \sigma(b,\tau)$ 
\cite{JPo, LiuQD23, ShamisSodinQD}.
For skew-shift models, the story is even more delicate (see \cite{HanJitomirskaya, JPo} for the Schr\"odinger case and \cite{ShamisSodinQD} for the long-range case). For $v$ real analytic and $\omega \in DC(\gamma, \tau)$,  there is $\lambda_0 = \lambda_0(v,\omega)$ such that 
\begin{equation*}
	\langle| X_{H_{x,\omega}}|_\phi^p\rangle(T) \leqslant C (\log T)^{p\sigma +\varepsilon}
\end{equation*}
holds uniformly in $x$ for $\lambda > \lambda_0$, with
\begin{equation*}
	\sigma = 4^{b - 1}b^3\tau^2.
\end{equation*}

Quasi-periodic and skew-shift models have been studied extensively; localization (see \cite{BGLocal, MetalInsulator} and references therein for quasi-periodic results and \cite{BourgainSkew, BGSSkew, HanLemmSchlag1, HanLemmSchlag2, SilviusSkew} and references therein for skew-shift results) and quantum dynamics (see \cite{BJBand, Fillman, GeKachkovskiy, JitomirskayaZhang,  LiuQD23, ZhangZhao, Zhao} and references therein for quasi-periodic results and \cite{HanJitomirskaya, JPo, ShamisSodinQD} and references therein for skew-shift results) are of particular interest.

The main idea of \cite{HanJitomirskaya} and \cite{JPo}, which are specifically applicable to Schr\"odinger operators with short-range potentials, involves the combination of an LDT (large deviation theorem) with a sublinear bound and a relation derived by Damanik and Tcheremchantsev \cite{DamanikTcherem1, JPo}. It was remarked in \cite{DamanikTcherem1, JPo} that the quantum dynamical bounds could be improved if either the LDT or the sublinear bound were improved. The idea of \cite{ShamisSodinQD}, which applies to long-range operators of the form we consider here, was to combine an LDT with integration along a suitable contour.

Here, we will take an approach which is closer in spirit to Liu \cite{LiuQD23}. There, one of the authors developed a new method inspired by Anderson localization proofs for quasi-periodic Schr\"odinger operators to show that, for long-range quasi-periodic operators, quantum dynamical upper bounds follow from a suitable sublinear bound of the bad Green's functions. We extend this argument to the skew-shift setting and prove novel sublinear bounds for the skew-shift to obtain better quantum dynamical upper bounds.

Denote
\begin{equation*}
	\psi(b)=\begin{cases}
	2^{b-1},& \text{if} \ 2\leqslant b\leqslant 5,\\
	b(b-1),&\text{if} \ b\geqslant 6.
	\end{cases}
\end{equation*}
We prove the following:

\begin{theorem}\label{MainTheorem}
	Suppose $\omega \in DC(\gamma,\tau)$. Let 
	\begin{equation*}
		(H_{x,\omega}u)_n = \sum_{n' \in \mathbb{Z}} A(n,n')u_{n'} +  v(f^nx) u_n
	\end{equation*}
	where $x \in \T^b$, $v$ is real analytic on $\mathbb{T}^b$, and $f$ is the skew-shift on $\mathbb{T}^b.$ Suppose $H_{x,\omega}$ satisfies the LDT (see Section \ref{App} for the precise definition). Then for any $\phi$ with compact support and $p > 0$ there exists $C=C(\varepsilon, v, A, b, \gamma, \tau, \phi, p)$ such that
	\begin{align}\label{QD1}
		\langle|\tilde X _{H_{x,\omega}}|_\phi^p\rangle(T) & \leqslant C (\log T)^{\frac{p}{\delta} + \varepsilon},            \\
		\langle| X _{H_{x,\omega}}|_\phi^p\rangle(T)       & \leqslant C (\log T)^{\frac{p}{\delta} + \varepsilon},\label{QD2} 
	\end{align}
	where $\delta =  \frac{1}{\tau b \psi(b)}$.
\end{theorem}

\begin{corollary}\label{Maincorollary}
	Suppose $\omega \in DC(\gamma,\tau)$. Let 
	\begin{equation*}
		(H_{x,\omega}u)_n = \sum_{n' \in \mathbb{Z}} A(n,n')u_{n'} +  \lambda v( f^nx) u_n
	\end{equation*}
	where $x\in \mathbb{T}^b$, $v$ is real analytic on $\mathbb{T}^{b}$, and $f$ is the skew-shift on $\T^b$. Then there is $\lambda_0 = \lambda_0(v,\omega)$ such that if $\lambda > \lambda_0$,  \eqref{QD1} and \eqref{QD2} hold.
\end{corollary}

\begin{proof}
	By Theorem 3.14 from \cite{LiuPDE}, there is $\lambda_0 = \lambda_0(v,\omega)$ such that the LDT holds for $\lambda>\lambda_{0}$. The corollary now follows immediately from Theorem \ref{MainTheorem}.
\end{proof}

For long-range operators, as far as we know, the previous best upper bound on $\langle| X _{H_{x,\omega}}|_\phi^p\rangle(T)$ in our framework, due to Shamis and Sodin \cite{ShamisSodinQD}, was 
$(\log T)^{p/\delta+ \varepsilon}$ with $\delta=(4^{b - 1}b^3\tau^2)^{-1}$. Here we have tightened this estimate to
$$(\log T)^{p \tau b\psi(b)+ \varepsilon} < (\log T)^{p 4^{b - 1}b^3\tau^2 + \varepsilon}.$$

For Schr\"odinger operators, to the best of our knowledge, our upper bounds are also best. Han and Jitomirskaya \cite{HanJitomirskaya}
obtained a bound of the discrepancy $N^{-1/(\tau  (2^{b}-1))+\varepsilon}$. From the discrepancy, it is possible to derive a sublinear bound $N^{1-\delta}$ (combining with the dimension $b$ loss) with $\delta=(\tau b (2^{b}-1))^{-1}$. Plugging this into the machinery developed by Liu \cite{LiuQD23} (see also Section \ref{App}), a weaker bound $(\log T)^{p/\delta+\varepsilon}$ than that in Theorem \ref{MainTheorem} arises. Later, Jitomirskaya and Powell \cite{JPo} improved the upper bound with $\delta=(\tau b^2 (2^{b-1}))^{-1}$ (in \cite{JPo}, a similar Diophantine condition was used with $\tau=1+\varepsilon$). According to the following relation:
\begin{equation*}
	\tau b\psi(b) \leqslant \min\left\{\tau b^{2} 2^{b-1} \ \text{\cite{JPo}}, \tau b(2^b -1) \ \text{\cite{HanJitomirskaya}}\right\},
\end{equation*}
Theorem \ref{MainTheorem} and Corollary \ref{Maincorollary} are, to the best of our knowledge, better than any previous bounds.

Let us now say a few words about our argument. Our central object is the so-called sublinear bound: suppose $\mathcal{S} \subseteq \T^b$ is a semi-algebraic set of degree $B$, then we say $\mathcal{S}$ satisfies a sublinear bound if for any $N \in \mathbb{N}$ with $\log B \lesssim \log N < - \log(\mathrm{Leb}(\mathcal{S}))$, the following inequality
\begin{equation*}
	\# \{1 \leqslant n \leqslant N: f^n x \in \mathcal{S}\} \leqslant N^{1 - \delta}
\end{equation*}
holds for all $x\in \T^b$.

First, we establish an abstract result (i.e. Theorem \ref{thm1}) relating sublinear bounds to upper bounds on $\langle|X_{H_{x,\omega}}|_\phi^p\rangle(T)$ by essentially following the argument from \cite{LiuQD23} (see Section \ref{App} for details). This method reduces the problem to proving a sublinear bound and is also a major source of improvement over the works of Han-Jitomirskaya, Jitomirskaya-Powell, and Shamis-Sodin \cite{HanJitomirskaya, JPo, ShamisSodinQD}. In fact, combining this result with the sublinear bound proved in \cite{JPo} or in \cite{LiuPDE} yields an improvement.

Our next step is to improve on the sublinear bounds established in \cite{JPo} and \cite{LiuPDE}. One fruitful way to obtain sublinear bounds on $\mathcal{S}\subseteq \mathbb{T}^{b}$ is to cover $\mathcal{S}$ by $\epsilon$-balls $B_\epsilon$ and estimate the following
\begin{equation*}
	\# \{1 \leqslant n \leqslant N: f^n x \in B_\epsilon\}.
\end{equation*}
This can be reduced, via Fourier analysis (see Lemma \ref{fejer}) to estimating an exponential sum of the form
\begin{equation*}
	\sum_{ |k_{1}| <R} \cdots \sum_{ |k_{b}| <R} \sum_{n = 1}^N \mathrm{e}^{2\pi\mathrm{i} \langle \mathbf{k}, f^{n}x \rangle}.
\end{equation*}
We employ two different number-theoretic arguments to estimate these exponential sums which are optimal in different situations. Specifically, we use the classic method of Weyl's method (see, e.g. \cite{hughbook}) when considering $\T^b, 2 \leqslant b \leqslant 5$ (see Section \ref{Weyl} for details), and we use Vinogradov's method (see, e.g. \cite{hughbook}) and recent proof of Vinogradov's mean value theorem by Bourgain-Demeter-Guth \cite{BDG16} when $b > 5$ (see Section \ref{Vino} for details). We refer readers to \cite{gz19} for the history and recent developments related to solutions of the Vinogradov system.
Han and Jitomirskaya also employed Weyl's method in \cite{HanJitomirskaya} to derive their upper bounds, but we introduce some techniques (square lattice decomposition) which improve on those estimates (c.f.  \cite[Section 5]{HanJitomirskaya} and Section \ref{Weyl} below).

The rest of our paper is organized as follows. In Section \ref{Preliminaries}, we provide useful definitions and prove estimates which are used throughout this paper. In Section \ref{Weyl}, we use Weyl's method to obtain a sublinear bound for semi-algebraic sets. In Section \ref{Vino}, we use the Vinogradov's method to obtain a sublinear bound for semi-algebraic sets. Finally, in Section \ref{App}, we detail how to relate a particular sublinear bound to an upper bound on $\langle|X _{H_{x,\omega}}|_\phi^p\rangle(T)$ and prove Theorem \ref{MainTheorem}.

\section{Preliminary}\label{Preliminaries}

We use $\|\cdot\|_{\mathbb{T}^{b}}$ to represent the distance to the nearest integer lattice in $\mathbb{Z}^{b}$. For $\mathbf{k}\in \mathbb{Z}^{b}$, define $\|\mathbf{k}\|:=\max\{|k_{1}|, |k_{2}|, \cdots, |k_{b}|\}$. Throughout the paper, $A\lesssim B$ means $A \leqslant C B$ for some constant $C>0$.  Denote by $\deg(P)$ the degree of the polynomial $P$.

\subsection{Exponential sums: Weyl's method}
Denote
\begin{equation*}
	S=S(\alpha)=\sum_{n=1}^{N}\mathrm{e}^{2\pi \mathrm{i} P(n;\alpha)},
\end{equation*}
where $P(x;\alpha)=\sum_{j=0}^{b}\alpha_{j}x^{j}$ is a polynomial with real coefficients.

\begin{lemma}\cite[page 42]{hughbook}\label{ddim}
	For any $b\geqslant 2$, 
	\begin{equation*}
		|S|^{2^{b-1}} \lesssim N^{2^{b-1}-1}+ N^{2^{b-1}-b} \sum_{h_{1}=1}^{N}\cdots \sum_{h_{b-1}=1}^{N}\min \bigg(N, \frac{1}{\|b!h_{1}\cdots h_{b-1}\alpha_{b}\|_{\mathbb{T}}}\bigg).
	\end{equation*}
\end{lemma}
\begin{remark}\label{ddim2}
	Note that $\min (N, \frac{1}{\| b!h_{1}\cdots h_{b-1}\alpha_{b}\|_{\mathbb{T}}})\geqslant 1$, thus Lemma \ref{ddim} implies
	\begin{equation*}
		|S|^{2^{b-1}} \lesssim N^{2^{b-1}-b} \sum_{h_{1}=1}^{N}\cdots \sum_{h_{b-1}=1}^{N}\min \bigg(N, \frac{1}{\|b!h_{1}\cdots h_{b-1}\alpha_{b}\|_{\mathbb{T}}}\bigg).
	\end{equation*}
\end{remark}

Multi-sums often appear in this paper. The following estimate is used frequently to reduce the multi-sums to a single sum.
\begin{lemma}\cite[Lemma 13, page 71]{Koro}\label{sums}
	Let $M$ and $m_{1},\cdots,{m}_{n}$ be positive integers. Denote by $\tau_{n}(M)$ the number of
	solutions of the equation $m_{1}\cdots m_{d}=M$. Then for any $0<\varepsilon\leqslant 1$,
	\begin{equation*}
		\tau_{n}(M) \leqslant C(\varepsilon, n)M^{\varepsilon}.
	\end{equation*}
\end{lemma}

\subsection{Exponential sums: Vinogradov's method}
Next, we recall the Vinogradov's Mean Value Theorem. It is obvious that $|S|$ is independent of $\alpha_{0}$. We put $\alpha=(\alpha_{1},\cdots,\alpha_{b})\in\mathbb{T}^{b}$ and reset $P(x;\alpha)=\sum_{j=1}^{b}\alpha_{j}x^{j}$. Denote
\begin{equation*}\label{Jdef1}
	J_{b}(N;\rho)=\int_{\mathbb{T}^{b}} |S(\alpha)|^{2\rho} \ \mathrm{d}\alpha.
\end{equation*}
We see that $J_{b}(N;\rho)$ is the number of solutions of the systems
\begin{equation*}\label{Jdef2}
	\begin{split}
		&m_{1}+\cdots+m_{\rho}=n_{1}+\cdots+n_{\rho},\\
		&m_{1}^{2}+\cdots+m_{\rho}^{2}=n_{1}^{2}+\cdots+n_{\rho}^{2},\\
		&\qquad \qquad \cdots\\
		&m_{1}^{b}+\cdots+m_{\rho}^{b}=n_{1}^{b}+\cdots+n_{\rho}^{b},
	\end{split}
\end{equation*}
where $1\leqslant m_{j},n_{j}\leqslant N$.

\begin{lemma}[Vinogradov's Mean Value Theorem]\cite[Theorem 1.1]{BDG16}\label{BDG}
	For each $\rho \geqslant 1$ and $b,N\geqslant 2$, the following upper bound holds:
	\begin{equation*}
		J_{b}(N;\rho)\leqslant C(\varepsilon, b, \rho) (N^{\rho+\varepsilon}+N^{2\rho-\frac{b(b+1)}{2}+\varepsilon} ).
	\end{equation*}
\end{lemma}

\begin{lemma}\cite[pages 79--81]{hughbook}\label{S2b}
	For any $\rho \geqslant 1$ and $b\geqslant 3$,
	\begin{equation*}
		|S|^{2\rho}\lesssim N^{\frac{(b-1)(b-2)}{2}-1} J_{b-1}(3N;\rho) \sum_{|h|\leqslant 2\rho b N^{b-1}} \min \bigg(N, \frac{1}{\| h\alpha_{b}\|_{\mathbb{T}}}\bigg).
	\end{equation*}
\end{lemma}
\begin{remark}
	For readers' convenience, we provide the proof of Lemma \ref{S2b} in the Appendix \ref{A1}.
\end{remark}

\subsection{Diophantine condition}
\begin{lemma}\label{keepDC}
	Let $\alpha\in DC(\gamma,\tau)$ and $p/q\in\mathbb{Q}$. There exists $\tilde{\gamma}=\tilde{\gamma}(p,q,\gamma)$ such that $\alpha p/q\in DC(\tilde{\gamma}, \tau)$.
\end{lemma}

\begin{lemma}\label{qn}
	Suppose that $\alpha\in DC(\gamma,\tau)$. Let $\{p_{n}/q_{n}\}$ be the best approximation of $\alpha$. Then for any $N\in\mathbb{N}$, there exists $q_{n}$ such that
	\begin{equation*}
		(\gamma N)^{\frac{1}{\tau}}<q_{n}\leqslant N.
	\end{equation*}
\end{lemma}
\begin{proof}
	On the one hand, for any $N\in\mathbb{N}$, there exist $q_{n}$ and $q_{n+1}$ such that
	\begin{equation*}
		q_{n}\leqslant N<q_{n+1}.
	\end{equation*}
	On the other hand, since $\alpha\in DC(\gamma, \tau)$, then
	\begin{equation*}
		\frac{\gamma}{q_{n}^{\tau}}\leqslant \|q_{n}\alpha\|_{\mathbb{T}} < \frac{1}{q_{n+1}},
	\end{equation*}
	which gives 
	\begin{equation*}
		q_{n}> (\gamma q_{n+1})^{\frac{1}{\tau}}> (\gamma N)^{\frac{1}{\tau}}.
	\end{equation*}
\end{proof}

The following estimate of the sum will be used frequently.
\begin{lemma}\label{dcweyl}
	Suppose that $\alpha\in DC(\gamma, \tau)$. Then
	\begin{equation*}
		\sum_{k=1}^{H}\min\bigg(N,\frac{1}{\|k\alpha\|_{\mathbb{T}}}\bigg)\leqslant \gamma^{-\frac{1}{\tau}} HN^{1-\frac{1}{\tau}} +H\log N+N+N\log N.
	\end{equation*}
\end{lemma}
\begin{proof}
	Let us first recall the following classic estimate (c.f. \cite[page 41]{hughbook}).
	If $|\alpha-p/q|\leqslant 1/q^{2}$, then
	\begin{equation}\label{weyl}
		\sum_{k=1}^{H}\min\bigg(N,\frac{1}{\|k\alpha\|_{\mathbb{T}}}\bigg)\leqslant \frac{HN}{q}+H\log q+N+q\log q.
	\end{equation}
	Now let $\{p_{n}/q_{n}\}$ be the best approximation of $\alpha$. It is well-known that
	\begin{equation*}
		\|q_{n}\alpha\|_{\mathbb{T}} <\frac{1}{q_{n+1}}<\frac{1}{q_{n}}, \ \text{for any} \ n\geqslant 1,
	\end{equation*}
	and thus $|\alpha-p_{n}/q_{n}|\leqslant 1/q_{n}^{2}$. Then the proof of Lemma \ref{dcweyl} is finished by combining  \eqref{weyl} with Lemma \ref{qn}.
\end{proof}

\subsection{Semi-algebraic set}

\begin{definition}[Semi-algebraic set]
	We say $\mathcal{S}\subseteq \mathbb{R}^{b}$ is a semi-algebraic set if it is a finite union of sets defined by a finite number of polynomial inequalities. More precisely, let $\{P_{1}, P_{2}, \cdots, P_{s}\}$ be a family of real polynomials to the variables $x=(x_{1}, x_{2}, \cdots, x_{b})$ with $\deg(P_{i})\leqslant d$ for $i=1,2,\cdots, s$. A (closed) semi-algebraic set $\mathcal{S}$ is given by the expression
	\begin{equation}\label{sas}
		\mathcal{S}=\cup_{j} \cap_{\ell\in \mathcal{L}_{j}} \{x\in \mathbb{R}^{b}: P_{\ell}(x) \ \varsigma_{j\ell} \ 0\},
	\end{equation}
	where $\mathcal{L}_{j}\subseteq \{1,2,\cdots, s\}$ and $\varsigma_{j\ell}\in \{\geqslant, \leqslant, =\}$. Then we say that the degree of $\mathcal{S}$, denoted by $\deg(\mathcal{S})$, is at most $sd$. In fact, $\deg(\mathcal{S})$ means the smallest $sd$ overall representation as in \eqref{sas}.
\end{definition}
\begin{lemma}\cite[Corollary 9.6]{Bou05}\label{covers}
	Let $\mathcal{S}\subseteq [0,1]^{b}$ be a semi-algebraic set of degree $B$. Let $\epsilon>0$ be a small number and $\mathrm{Leb}(\mathcal{S})\leqslant \epsilon^{b}$. Then $\mathcal{S}$ can be covered by a family of $\epsilon$-balls with total number less than $B^{C(b)} \epsilon^{1-b}$.
\end{lemma}

\subsection{Fourier analysis}

Let $\chi_{\epsilon}(\cdot)$ be the characteristic function of the ball in $\mathbb{T}^{b}$ of radius $\epsilon$ centered at $0$. That is,
\begin{equation*}
	\chi_{\epsilon}(x)=
	\begin{cases}
		1, & \|x\|_{\mathbb{T}^{b}}<\epsilon, \\
		0, & \text{others}.                   
	\end{cases}
\end{equation*}
The following result is a part of the calculation in \cite{JPo}. For convenience, we state it as a lemma and provide the proof.
\begin{lemma}\cite[pages 189--190]{JPo}\label{fejer}
	Suppose that $\{x_{n}\}_{n=1}^{N}\subseteq \mathbb{T}^{b}$. Let $R=[\epsilon^{-1}/10]$. Then
	\begin{equation*}
		\sum_{n=1}^{N}\chi_{\epsilon}(x_{n})\leqslant C(b) R^{-b}\sum_{\|\mathbf{k}\| <R} \bigg| \sum_{n=1}^{N}\mathrm{e}^{2\pi \mathrm{i} \langle \mathbf{k}, x_{n}\rangle } \bigg|.
	\end{equation*}
\end{lemma}
\begin{remark}
    In this lemma and the rest of this paper, $[x]$ denotes the integer part of $x\in \R$. 
\end{remark}
\begin{proof}
	Let $F(\cdot)$ be the usual Fej\'er kernel on $\mathbb{R}$:
	\begin{equation*}
		F(x)=\frac{1}{R}\bigg(\frac{\sin \pi Rx}{\sin \pi x}\bigg)^{2}=\sum_{|k|<R}\bigg(1-\frac{|k|}{R}\bigg)\mathrm{e}^{2\pi \mathrm{i}kx}=:\sum_{|k|<R}\widehat{F}(k)\mathrm{e}^{2\pi \mathrm{i}kx}.
	\end{equation*}
	For any $x=(x_{1},\cdots,x_{b})\in\mathbb{T}^{b}$, if $\chi_{\epsilon}(x)=0$, then $\chi_{\epsilon}(x)\leqslant R^{-b}\prod_{j=1}^{b}F(x_{j})$ holds trivially. On the other hand, if $\chi_{\epsilon}(x)=1$, then $R/2 \leqslant F(x_{j})\leqslant 2R$ for every $1\leqslant j\leqslant b$, and we also have $\chi_{\epsilon}(x)\leqslant 2^{b}R^{-b}\prod_{j=1}^{b}F(x_{j})$. Thus
	\begin{equation*}
		\chi_{\epsilon}(x)\lesssim R^{-b}\prod_{j=1}^{b}F(x_{j}), \ \text{for any} \ x\in \mathbb{T}^{b}.
	\end{equation*}
	Set $\mathbf{k}=(k_{1},\cdots, k_{b})\in \mathbb{Z}^{b}$ and let $\widehat{F}(\mathbf{k})=\prod_{j=1}^{b}\widehat{F}(k_{j})$. Then
	\begin{equation*}
		\begin{split}		
			\prod_{j=1}^{b}F(x_{j})&=\prod_{j=1}^{b}\sum_{|k_{j}|<R}\widehat{F}(k_{j})\mathrm{e}^{2\pi \mathrm{i}k_{j}x_{j}}\\
			&=\sum_{\|\mathbf{k}\| <R}\widehat{F}(\mathbf{k}) \mathrm{e}^{2\pi \mathrm{i}\langle \mathbf{k},x\rangle}.
		\end{split}
	\end{equation*}
	Note that $|\widehat{F}(\mathbf{k})|\leqslant \prod_{j=1}^{b}|\widehat{F}(k_{j})|\leqslant 1$, then
	\begin{equation*}
		\sum_{n=1}^{N} \chi_{\epsilon}(x_{n}) \lesssim R^{-b} \sum_{n=1}^{N} \sum_{\|\mathbf{k}\| <R} \widehat{F}(\mathbf{k}) \mathrm{e}^{2\pi \mathrm{i} \langle \mathbf{k}, x_{n} \rangle} \lesssim  R^{-b}\sum_{\|\mathbf{k}\| <R} \bigg| \sum_{n=1}^{N}\mathrm{e}^{2\pi \mathrm{i} \langle \mathbf{k}, x_{n}\rangle } \bigg|.
	\end{equation*}
	This finishes the proof.
\end{proof}

\section{Sublinear bound: Weyl's method}\label{Weyl}

\begin{theorem}\label{wm}
	Let $\mathbf{P}=(P_{1},\cdots, P_{b})$ be a vector $(b\geqslant 2)$ of real polynomials and $\alpha_{i}$ be the leading coefficient of $P_{i}$ for $1\leqslant i\leqslant b$. Suppose $1\leqslant \deg(P_{1})  <\cdots <\deg(P_{b}) \leqslant m$ and $\alpha_{i} \in DC(\gamma, \tau)$ with $\gamma >0,\tau>1$ for $1\leqslant i\leqslant b$.  Let $\mathcal{S}\subseteq [0,1]^{b}$ be a semi-algebraic set of degree $B$ and $\mathrm{Leb}(\mathcal{S})<\eta$. Let $N\in\mathbb{N}$ such that 
	\begin{equation*}
		\log B \lesssim\ \log N< 2^{m-1} \tau \log \frac{1}{\eta}.
	\end{equation*}
	Then for any $0<\varepsilon\leqslant 1$,
	\begin{equation*}
		\#\{1\leqslant n\leqslant N: \mathbf{P}(n)\bmod \mathbb{Z}^{b}\in\mathcal{S}\} \leqslant C(\varepsilon,b,m,\gamma,\tau)B^{C(b)}N^{1-\frac{1}{\tau b 2^{m-1}}+\varepsilon}.
	\end{equation*}
\end{theorem}
\begin{proof}
	Let $\epsilon=N^{-\frac{1}{\tau b 2^{m-1}}}$ and $R=[\epsilon^{-1}/10]$.
	Then
	\begin{equation*}
		\mathrm{Leb}(\mathcal{S})< \eta < N^{-\frac{1}{\tau 2^{m-1}}}=\epsilon^{b}.
	\end{equation*}
	Thus by Lemma \ref{covers}, $\mathcal{S}$ can be covered by at most $B^{C(b)}$ $\epsilon^{1-b}$ many $\epsilon$-balls.
		
	Consider one such ball; without loss of generality, we may assume the ball is centered at 0. Let $\chi_{\epsilon}(\cdot)$ be the characteristic function of the ball in $\mathbb{T}^{b}$ with radius $\epsilon$ centered at $0$. 
	Apply Lemma \ref{fejer} we have
	\begin{equation}\label{domdd}
		\sum_{n=1}^{N} \chi_{\epsilon}(\mathbf{P}(n)\bmod \mathbb{Z}^{b}) 
		\lesssim R^{-b}\sum_{\|\mathbf{k}\| <R} \bigg| \sum_{n=1}^{N}\mathrm{e}^{2\pi \mathrm{i} \langle \mathbf{k}, \mathbf{P}(n)\rangle } \bigg|.
	\end{equation}
	To estimate \eqref{domdd}, we decompose the  square lattice $\{\mathbf{k}: \|\mathbf{k}\| <R\}$ such that
	\begin{equation*}
		\{\mathbf{k}: \|\mathbf{k}\| <R\} = \cup_{i} \mathcal{K}^{i}, \ \text{and}\ \mathcal{K}^{i}\cap \mathcal{K}^{i'}=\emptyset, \ \text{for any}\ i\neq i',
	\end{equation*}
	where
	\begin{equation}\label{K}
		\begin{split}
			&\mathcal{K}^{b}=\{\|\mathbf{k}\| <R: k_{b}\neq 0\},\\
			&\mathcal{K}^{i}=\{\|\mathbf{k}\| <R: k_{i}\neq 0, k_{i'}= 0 \ \text{for any}\ i'>i\}, \ 1\leqslant i<b,\\
			&\mathcal{K}^{0}=\{k_{b}=k_{b-1}=\cdots =k_{1}=0\}.
		\end{split}
	\end{equation}
	Thus it follows from \eqref{domdd} and \eqref{K} that
	\begin{equation}\label{decsum}
		\sum_{n=1}^{N} \chi_{\epsilon}(\mathbf{P}(n)\bmod \mathbb{Z}^{b}) 
		\lesssim NR^{-b}+ R^{-b}\sum_{i=1}^{b} \sum_{\mathbf{k}\in \mathcal{K}^{i}} \bigg| \sum_{n=1}^{N}\mathrm{e}^{2\pi \mathrm{i} \langle \mathbf{k}, \mathbf{P}(n)\rangle } \bigg|.
	\end{equation}
	Denote 
	\begin{equation*}
		S=\sum_{n=1}^{N}\mathrm{e}^{2\pi \mathrm{i} \langle \mathbf{k}, \mathbf{P}(n)\rangle }.
	\end{equation*}
	Let us consider one such $\mathcal{K}^{i}$. Denote $\deg (P_{i}) = D_{i}$. We need to deal with two cases:

	{\bf Case 1: $D_{i}=1$.}
	Obviously, by the monotonicity of $D_{i}$ to $i$ and $1\leqslant D_{i}\leqslant m$, it must holds that $i=1$. Then for $\mathbf{k}\in \mathcal{K}^{1}$,
	\begin{equation*}
		\langle \mathbf{k},\mathbf{P}(n)\rangle =k_{1}\alpha_{1} n+O(1).
	\end{equation*}
	Combining the estimate of the geometric series with Lemma \ref{dcweyl} shows
	\begin{equation}\label{Di1}
		\sum_{\mathbf{k}\in \mathcal{K}^{1}} |S| \leqslant 2\sum_{k_{1}=1}^{R} \min \bigg(N, \frac{1}{\|k_{1}\alpha_{1}\|}\bigg) \lesssim RN^{1-\frac{1}{\tau}}\lesssim R^{b} N^{1-\frac{1}{\tau 2^{m-1}}}.
	\end{equation}

	{\bf Case 2: $D_{i}\geqslant 2$.}
	By the assumption of $\mathbf{P}$ and $\mathbf{k}\in \mathcal{K}^{i}$, we know
	\begin{equation*}
		\langle \mathbf{k},\mathbf{P}(n)\rangle=k_{i}\alpha_{i}n^{D_{i}}+O(n^{D_{i}-1}).
	\end{equation*}
	Since $\# \mathcal{K}^{i} \lesssim R^{i}$, by the H\"older inequality,
	\begin{equation}\label{holderdd}
		\sum_{\mathbf{k}\in \mathcal{K}^{i}}|S| \leqslant R^{i(1-\frac{1}{2^{D_{i}-1}})} \bigg(\sum_{\mathbf{k}\in\mathcal{K}^{i}} |S|^{2^{D_{i}-1}}\bigg)^{1/2^{D_{i}-1}}.
	\end{equation}
	By Lemma \ref{ddim} and Remark \ref{ddim2}, we have
	\begin{equation*}
		|S|^{2^{D_{i}-1}} \lesssim N^{2^{D_{i}-1}-D_{i}} \sum_{h_{1}=1}^{N}\cdots \sum_{h_{D_{i}-1}=1}^{N}\min \bigg(N, \frac{1}{\|(D_{i})!h_{1}\cdots h_{D_{i}-1}k_{i}\alpha_{i}\|_{\mathbb{T}}}\bigg).
	\end{equation*}
	We apply Lemma \ref{sums} to reduce the multi-sums to a single sum: for any $0<\varepsilon\leqslant 1$,
	\begin{equation*}
		|S|^{2^{D_{i}-1}} \lesssim N^{2^{D_{i}-1}-D_{i}+\varepsilon} \sum_{h=1}^{(D_{i})! N^{D_{i}-1}} \min \bigg(N, \frac{1}{\|hk_{i}\alpha_{i}\|_{\mathbb{T}}}\bigg).
	\end{equation*}
	Similarly, combining the above inequality with Lemma \ref{sums} again, so that
	\begin{equation}\label{Sldd}
		\begin{split}	\sum_{\mathbf{k}\in\mathcal{K}^{i}}|S|^{2^{D_{i}-1}}&\lesssim N^{2^{D_{i}-1}-D_{i}+\varepsilon} R^{i-1} \sum_{k_{i}=1}^{R}\sum_{h=1}^{(D_{i})! N^{D_{i}-1}}\min \bigg(N, \frac{1}{\| h k_{i}\alpha_{i}\|_{\mathbb{T}}}\bigg)\\
			&\lesssim N^{2^{D_{i}-1}-D_{i}+\varepsilon} R^{i-1} \sum_{k=1}^{(D_{i})!RN^{D_{i}-1}}\min \bigg(N, \frac{1}{\|k\alpha_{i}\|_{\mathbb{T}}}\bigg).
		\end{split}
	\end{equation}
	Since $\alpha_{i}\in DC(\gamma, \tau)$ with $\tau>1$, we apply Lemma \ref{dcweyl} with $H=(D_{i})!RN^{D_{i}-1}$, 
	\begin{equation}\label{fractiondd}
		\sum_{k=1}^{(D_{i})!RN^{D_{i}-1}}\min \bigg(N, \frac{1}{\|k\alpha_{i}\|_{\mathbb{T}}}\bigg)\lesssim RN^{D_{i}-\frac{1}{\tau}}.
	\end{equation}
	Substituting \eqref{fractiondd} into \eqref{Sldd} shows that
	\begin{equation}\label{Sdd}
		\sum_{\mathbf{k}\in \mathcal{K}^{i}}|S|^{2^{D_{i}-1}}\lesssim  R^{i} N^{2^{D_{i}-1}-\frac{1}{\tau}+\varepsilon}.
	\end{equation}
	Now we combine \eqref{Sdd} with \eqref{holderdd}, one can see that
	\begin{equation}\label{Di2}
		\begin{split}
			\sum_{\mathbf{k}\in \mathcal{K}^{i}}|S| &\lesssim  R^{i(1-\frac{1}{2^{D_{i}-1}})} (R^{i} N^{2^{D_{i}-1}-\frac{1}{\tau}+\varepsilon} )^{\frac{1}{2^{D_{i}-1}}}\\
			&\lesssim R^{i} N^{1-\frac{1}{\tau 2^{D_{i}-1}}+\varepsilon} \lesssim  R^{b}N^{1-\frac{1}{\tau 2^{m-1}}+\varepsilon},  
		\end{split}
	\end{equation}
	where the last step uses the assumption that $D_{i}\leqslant m$.
	 
	Thus, {\bf Case 1} and {\bf Case 2} imply that
	\begin{equation*}
		\sum_{\mathbf{k}\in \mathcal{K}^{i}}|S| \lesssim  R^{b}N^{1-\frac{1}{\tau 2^{m-1}}+\varepsilon}, \ \text{for any}\ 1\leqslant i\leqslant m.
	\end{equation*}
	Combine the above inequality with \eqref{decsum}, we have
	\begin{equation*}
		\begin{split}
			\sum_{n=1}^{N} \chi_{\epsilon}(\mathbf{P}(n)\bmod\mathbb{Z}^{b})&\lesssim NR^{-b}+ R^{-b} \sum_{i=1}^{b}  R^{b} N^{1-\frac{1}{\tau 2^{m-1}}+\varepsilon}\\
			& \lesssim NR^{-b}+N^{1-\frac{1}{\tau 2^{m-1}}+\varepsilon}.
		\end{split}
	\end{equation*}
	According to the definition of $R$, we know $R^{-b}\lesssim N^{-\frac{1}{\tau 2^{m-1}}}$ and thus
	\begin{equation*}
		\sum_{n=1}^{N} \chi_{\epsilon}(\mathbf{P}(n)\bmod\mathbb{Z}^{b})\lesssim N^{1-\frac{1}{\tau 2^{m-1}}+\varepsilon}.
	\end{equation*}
	Since we only need $B^{C(b)}$ $\epsilon^{1-b}$ many $\epsilon$-balls to cover $\mathcal{S}$, we get 
	\begin{equation*}
		\begin{split}
			\#\{1\leqslant n\leqslant N: \mathbf{P}(n)\bmod\mathbb{Z}^{b}\in\mathcal{S}\} &\leqslant B^{C(b)} \epsilon^{1-b} \sum_{n=1}^{N} \chi_{\epsilon}(\mathbf{P}(n)\bmod\mathbb{Z}^{b}) \\
			&\leqslant C(\varepsilon,b,m,\gamma,\tau)B^{C(b)}N^{1-\frac{1}{\tau b 2^{m-1}}+\varepsilon}.
		\end{split}
	\end{equation*}
	This finishes the proof.
\end{proof}

\section{Sublinear bound: Vinogradov's method}\label{Vino}

\begin{theorem}\label{vm}
	Let $\mathbf{P}=(P_{1},\cdots, P_{b})$ be a vector $(b\geqslant 2)$ of real polynomials and $\alpha_{i}$ be the leading coefficient of $P_{i}$ for $1\leqslant i\leqslant b$. Suppose $1\leqslant \deg(P_{1})  <\cdots <\deg(P_{b}) \leqslant m$ and $\alpha_{i} \in DC(\gamma, \tau)$ with $\gamma >0,\tau>1$ for $1\leqslant i\leqslant b$.  Let $\mathcal{S}\subseteq [0,1]^{b}$ be a semi-algebraic set of degree $B$ and $\mathrm{Leb}(\mathcal{S})<\eta$. Let $N\in\mathbb{N}$ such that 
	\begin{equation*}
		\log B \lesssim \log N< m(m-1)\tau \log \frac{1}{\eta}.
	\end{equation*}
	Then for any $0<\varepsilon\leqslant 1$,
	\begin{equation*}
		\#\{1\leqslant n\leqslant N: \mathbf{P}(n)\bmod\mathbb{Z}^{b}\in\mathcal{S}\} \leqslant C(\varepsilon,b,m,\gamma,\tau)B^{C(b)}N^{1-\frac{1}{\tau b m(m-1)}+\varepsilon}.
	\end{equation*}
\end{theorem}
\begin{proof}
	Let $\epsilon=N^{-\frac{1}{\tau b m(m-1)}}$ and $R=[\epsilon^{-1}/10]$.
	Then
	\begin{equation*}
		\mathrm{Leb}(\mathcal{S}) < \eta < N^{-\frac{1}{\tau m(m-1)}}=\epsilon^{b}.
	\end{equation*}
	Thus by Lemma \ref{covers}, $\mathcal{S}$ can be covered by at most $B^{C(b)}$ $\epsilon^{1-b}$ many $\epsilon$-balls.	
		
	Consider one such ball; without loss of generality, we may assume the ball is centered at 0. Let $\chi_{\epsilon}(\cdot)$ be the characteristic function of the ball in $\mathbb{T}^{b}$ with radius $\epsilon$ centered at $0$. Apply Lemma \ref{fejer} we have
	\begin{equation}\label{vdomdd}
		\sum_{n=1}^{N} \chi_{\epsilon}(\mathbf{P}(n)\bmod\mathbb{Z}^{b}) \lesssim  R^{-b}\sum_{\|\mathbf{k}\| <R} \bigg| \sum_{n=1}^{N}\mathrm{e}^{2\pi \mathrm{i} \langle \mathbf{k}, \mathbf{P}(n)\rangle } \bigg|.
	\end{equation}
	Just as we did in Section \ref{Weyl}, the square lattice $\{\mathbf{k}: \|\mathbf{k}\| <R\}$ are decomposed as the disjoint sets $\mathcal{K}^{i}, 0\leqslant i\leqslant b$. That is, \eqref{vdomdd} can be rewritten as
	\begin{equation}\label{vdec}
		\sum_{n=1}^{N} \chi_{\epsilon}(\mathbf{P}(n)\bmod\mathbb{Z}^{b}) \lesssim  NR^{-b}+R^{-b}\sum_{i=1}^{b} \sum_{\mathbf{k}\in \mathcal{K}^{i}} \bigg| \sum_{n=1}^{N}\mathrm{e}^{2\pi \mathrm{i} \langle \mathbf{k}, \mathbf{P}(n)\rangle } \bigg|,
	\end{equation}
	where $\mathcal{K}^{i}$ is defined in \eqref{K}. Denote 
	\begin{equation*}
		S:=\sum_{n=1}^{N}\mathrm{e}^{2\pi \mathrm{i} \langle \mathbf{k}, \mathbf{P}(n)\rangle }.
	\end{equation*}
	Fix $i$ and consider one such $\mathcal{K}^{i}$. Denote $\deg(P_{i})=D_{i}$. We need to deal with two cases: 
	
	{\bf Case 1: $1\leqslant D_{i}\leqslant 2$.}
	As similar as \eqref{Di1} and \eqref{Di2}, the following holds:
\begin{equation*}
		\sum_{\mathbf{k}\in \mathcal{K}^{i}} |S| \lesssim R^{i}N^{1-\frac{1}{\tau 2^{D_{i}-1}} +\varepsilon }\lesssim R^{b} N^{1-\frac{1}{\tau m(m-1)} +\varepsilon}.
	\end{equation*}
	
	{\bf Case 2: $D_{i}\geqslant 3$.}
	By the assumption of $\mathbf{P}$ and $\mathbf{k}\in \mathcal{K}^{i}$, we know
	\begin{equation*}\label{kxnd}
		\langle \mathbf{k},\mathbf{P}(n)\rangle =k_{i}\alpha_{i} n^{D_{i}}+O(n^{D_{i}-1}).
	\end{equation*}
	Since $\# \mathcal{K}^{i} \lesssim R^{i}$, by the H\"older inequality, for any $\rho\in \mathbb{Z}^{+}$ ($\rho$ to be determined),
	\begin{equation}\label{vholderd}
		\sum_{\mathbf{k}\in \mathcal{K}^{i}} |S| \leqslant R^{i(1-\frac{1}{2\rho})} \bigg(\sum_{\mathbf{k}\in \mathcal{K}^{i}} |S|^{2\rho}\bigg)^{1/(2\rho)}.
	\end{equation}
	By Lemma \ref{S2b}, for $\mathbf{k}\in \mathcal{K}^{i}$,
	\begin{equation}\label{S2rho}
		|S|^{2\rho} \lesssim N^{\frac{(D_{i}-1)(D_{i}-2) }{2}-1} J_{D_{i}-1}(3N;\rho)  \sum_{|h|\leqslant 2\rho N^{D_{i}-1}}  \min \bigg(N, \frac{1}{\|hk_{i}\alpha_{i}\|_{\mathbb{T}}}\bigg).
	\end{equation}
	Notice that $\min (N, \frac{1}{\|hk_{i}\alpha_{i}\|_{\mathbb{T}}})\geqslant 1$, then the inner sum in \eqref{S2rho} satisfies that
	\begin{equation}\label{inner}
		\begin{split}
			\sum_{|h|\leqslant 2\rho N^{D_{i}-1}}  \min \bigg(N, \frac{1}{\|hk_{i}\alpha_{i}\|_{\mathbb{T}}}\bigg)  &= N + 2\sum_{h=1}^{2\rho N^{D_{i}-1}}  \min \bigg(N, \frac{1}{\|hk_{i}\alpha_{i}\|_{\mathbb{T}}}\bigg)\\
			&\lesssim \sum_{h=1}^{2\rho N^{D_{i}-1}}  \min \bigg(N, \frac{1}{\|hk_{i}\alpha_{i}\|_{\mathbb{T}}}\bigg),
		\end{split}  
	\end{equation}
	where the last step uses that $D_{i}\geqslant 3$. Combine \eqref{S2rho} and \eqref{inner} with  Lemma \ref{sums}, for any $0<\varepsilon\leqslant 1$, we have
	\begin{equation*}
		\sum_{\mathbf{k}\in \mathcal{K}^{i}}|S|^{2\rho} \lesssim  R^{i-1}N^{\frac{(D_{i}-1)(D_{i}-2)}{2}-1+\varepsilon}  J_{D_{i}-1}(3N;\rho) \sum_{k=1}^{2\rho RN^{D_{i}-1}}  \min \bigg(N, \frac{1}{\|k\alpha_{i}\|_{\mathbb{T}}}\bigg).  
	\end{equation*}
	Thus it follows from Lemma \ref{dcweyl} and $\alpha_{i}\in DC(\gamma, \tau)$ with $\tau>1$ that
	\begin{equation*}
		\sum_{\mathbf{k}\in \mathcal{K}^{i}}|S|^{2\rho } \lesssim R^{i} N^{\frac{D_{i}(D_{i}-1)}{2}-\frac{1}{\tau}+\varepsilon}J_{D_{i}-1}(3N;\rho).
	\end{equation*}
	Recall the Vinogradov's Mean Value Theorem (see Lemma \ref{BDG}):
	\begin{equation*}
		J_{D_{i}-1}(3N;\rho)\lesssim N^{\rho+\varepsilon}+N^{2\rho- \frac{D_{i}(D_{i}-1)}{2}+\varepsilon}.
	\end{equation*}
	Choose $2\rho=D_{i}(D_{i}-1)$. Then we have
	\begin{equation*}
		\sum_{\mathbf{k}\in \mathcal{K}^{i}} |S|^{2\rho}\lesssim R^{i}N^{D_{i}(D_{i}-1)-\frac{1}{\tau}+\varepsilon}.
	\end{equation*}
	Substituting the above inequality into \eqref{vholderd} shows
	\begin{equation*}
		\begin{split}
			\sum_{\mathbf{k}\in \mathcal{K}^{i}} |S| &\lesssim R^{i(1-\frac{1}{D_{i}(D_{i}-1)})} (R^{i}N^{D_{i}(D_{i}-1)-\frac{1}{\tau}+\varepsilon})^{\frac{1}{D_{i}(D_{i}-1)}}\\
			&\lesssim R^{i} N^{1-\frac{1}{\tau D_{i}(D_{i}-1)}+\varepsilon} \lesssim R^{b} N^{1-\frac{1}{\tau m(m-1)}+\varepsilon},
		\end{split}
	\end{equation*}
	where the last step uses the assumption that $D_{i}\leqslant m$.
	 
	Therefore, regardless of {\bf Case 1} or {\bf Case 2}, it always holds that
	\begin{equation*}
		\sum_{\mathbf{k}\in \mathcal{K}^{i}} |S| \lesssim R^{b} N^{1-\frac{1}{\tau m(m-1)}+\varepsilon}, \ \text{for any}\ 1\leqslant i\leqslant m.
	\end{equation*}
	Combine the above estimate with \eqref{vdec}, one can see that
	\begin{equation*}
		\begin{split}
			\sum_{n=1}^{N} \chi_{\epsilon}(\mathbf{P}(n)\bmod\mathbb{Z}^{b}) &\lesssim NR^{-b}+  R^{-b}\sum_{i=1}^{b} R^{b} N^{1-\frac{1}{\tau m(m-1)}+\varepsilon}\\
			&\lesssim NR^{-b}+ N^{1-\frac{1}{\tau m(m-1)}+\varepsilon}.
		\end{split}
	\end{equation*}
	According to the definition of $R$, we know $R^{-b}\lesssim N^{-\frac{1}{\tau m(m-1)}}$, thus
	\begin{equation*}
		\sum_{n=1}^{N} \chi_{\epsilon}(\mathbf{P}(n)\bmod\mathbb{Z}^{b})\lesssim  N^{1-\frac{1}{\tau m(m-1)}+\varepsilon}.
	\end{equation*}
	Finally, we cover $\mathcal{S}$ by $B^{C(b)} \epsilon^{1-b}$ many $\epsilon$-balls so that
	\begin{equation*}
		\begin{split}
			\#\{1\leqslant n\leqslant N: \mathbf{P}(n)\bmod\mathbb{Z}^{b}\in\mathcal{S}\}&\leqslant B^{C(b)} \epsilon^{1-b} \sum_{n=1}^{N} \chi_{\epsilon}(\mathbf{P}(n)\bmod\mathbb{Z}^{b})\\
			&\leqslant C(\varepsilon,b, m, \gamma,\tau)B^{C(b)}N^{1-\frac{1}{\tau b m(m-1)}+\varepsilon}.
		\end{split} 
	\end{equation*}	
	This finishes the proof.
\end{proof}

\section{Proof of Theorem \ref{MainTheorem}}\label{App}

\subsection{Large deviation, sublinear bound, and quantum dynamics}

We begin this section with an abstract result (i.e. Theorem \ref{thm1}) that applies to bounded long-range operators 
\begin{equation*}
	(Hu)_n = \sum_{n' \in \Z^d} A(n,n') u_{n'} + V(n)u_{n}
\end{equation*}
on  $\ell^2(\Z^d)$ which satisfy
\begin{enumerate}
	\item[(1)]
	      for any $n,n'\in\Z^d$,
	      \begin{equation*}\label{GO11}
	      	|A(n,n')|\leqslant   C_1 \mathrm{e}^{-c_1 \|n-n'\|},  C_1>0, c_1>0;
	      \end{equation*}
	\item[(2)]  
	      for any  $n,n'\in\Z^d$, 
	      \begin{equation*}\label{GO12}
	      	A(n,n')=\overline{A(n',n)};
	      \end{equation*}
	\item[(3)]		
	      for any $n, n', k\in \Z^d$, 
	      \begin{equation*}\label{gdec72}
	      	A(n+k,n^\prime+k)= A(n,n^\prime).
	      \end{equation*}
	      	      		
\end{enumerate}
Theorem \ref{thm1} may be of independent interest and could find applications in the future.
	
Let us explain the setting. We first recall the concept of elementary region. For $d=1$,   the elementary region of size $N$ centered at 0 is given by
\begin{equation*}
	Q_N=[-N,N].
\end{equation*}
For $d\geqslant 2$, denote by $Q_N$ an elementary region of size $N$ centered at 0, which is one of the following regions,
\begin{equation*}
	Q_N=[-N,N]^d
\end{equation*}
or
$$Q_N=[-N,N]^d\setminus\{n\in\mathbb{Z}^d: \ n_i \ \varsigma_i \ 0, 1\leqslant i\leqslant d\},$$
where  for $ i=1,2,\cdots,d$, $ \varsigma_i\in \{<,>,\emptyset\}$ and at least two $ \varsigma_i$  are not $\emptyset$. Denote by $\mathcal{E}_N^{0}$ the set of all elementary regions of size $N$ centered at 0. 
Let
$$\mathcal{E}_N:=\{n+Q_N:n\in\mathbb{Z}^d,Q_N\in \mathcal{E}_N^{0}\}.$$
We call the elements in $\mathcal{E}_N$ elementary regions.

Next, we define the Green's function. Let  $R_{{\Lambda}}$  be the operator of restriction (projection) to $\Lambda \subseteq \Z^d$. Define the Green's function at $z$ by
\begin{equation*}\label{g0}
	G_{{\Lambda}}(z)=(R_{{\Lambda}}(H-zI)R_{{\Lambda}})^{-1}.
\end{equation*}
Set $ G(z)=(H-zI)^{-1}$.  Clearly, both $  G_{{\Lambda}}(z)$ and $G(z)$ are always well-defined for $z\in \C_+\equiv \{z\in \C: \Im z>0\}$. Fixed $0<\sigma_1<1$,
we say  an elementary region $ \Lambda\in \mathcal{E}_{N}$ is  in class {\it SG$_{N}$} (strongly good with size $N$) if
\begin{align*}
	\|G_{\Lambda}(z)\|           & \leqslant  \mathrm{e}^{N^{\sigma_1}},                                         \\
	|G_{\Lambda}(z)(n,n^\prime)| & \leqslant  \mathrm{e}^{-c_2 \|n-n' \|}, \text{ for } \|n-n' \|\geqslant N/10, 
\end{align*}

Finally, we note that the self-adjoint operator $H$  is bounded, so there exists a large $K>0$ such that $\sigma({H})\subseteq [-K+1,K-1]$. 
 
\begin{theorem}\cite[Corollary 2.3]{LiuQD23}\label{cor1}
	Define $\mathcal{B}_{N,N_1}$ as
	\begin{equation*}
		\mathcal{B}_{N,N_1}=\{ n\in [-N,N]^d: \text{ there exists } Q_{N_1}\in  \mathcal{E}_{N_1}^{0} \text{ such that } n+Q_{N_1}\notin \text{{\it SG}}_{N_1}\}.
	\end{equation*}
	Assume that there exists $\epsilon_0>0$ such that for any $z=E+\mathrm{i}\epsilon $ with  $|E|\leqslant K$ and $0<\epsilon\leqslant \epsilon_0$, and arbitrarily  small $\varepsilon>0$,
	\begin{equation}\label{gsublinear}
		\#  \mathcal{B}_{N,[ N^{\varepsilon} ]}\leqslant N^{1-\delta} {\text{ when }} N \geqslant N_0 
	\end{equation}
	($N_0$ may depend on $\varepsilon$).
	Then  for any $\phi$ with compact support and any $\varepsilon>0$ there exists $T_0>0$ (depending on $d,p,\phi,K$, $\sigma_{1},\delta, \epsilon_0$, $c_1$, $c_2$, $C_1$, $N_0$ and $\varepsilon$) such that for any $T\geqslant T_0$,
	\begin{align*}
		\langle|\tilde{X}_H|_{\phi}^p\rangle(T) & \leqslant  (\log T)^{  \frac{p}{\delta} +\varepsilon}, \\
		\langle|{X}_H|_{\phi}^p\rangle(T)       & \leqslant  (\log T)^{  \frac{p}{\delta} +\varepsilon}. 
	\end{align*}
\end{theorem}

Theorem \ref{cor1} can be applied to dynamically-defined operators. Let $f$ be a function from $\Z^d\times\T^b$ to $\T^b$. Assume for any $m,n\in\mathbb{Z}^{d}$,
\begin{equation*}
	f(m+n,x)= f(m, f(n,x)).
\end{equation*}
Sometimes, we write down $f^{n} x$ for $f(n,x)$ for convenience, where $n\in \Z^d$ and $x\in \T^b$. Define a family of  operators $\{H_x\}_{x\in\mathbb{T}^{b}}$  on $\ell^2(\Z^d)$:
\begin{equation*}\label{ops}
	(H_{x}u)_n = \sum_{n' \in \Z^d} A(n,n') u_{n'} + v(f(n,x))u_{n},
\end{equation*}
where  $v$ is a real  analytic   function on $\T^b$.

In fact, \eqref{gsublinear} in Theorem \ref{cor1} can be deduced from the assumption of the LDT (large deviation theorem) and the sublinear bound for the semi-algebraic set. More precisely, we say the Green's function  of an operator $H_x$  satisfies the LDT in complexified energies (sometimes just say LDT for short)  if  there exist $\epsilon_0>0$ and $N_0>0$ such that for  any  $N\geqslant N_0$, there exists a  subset $X_N\subseteq \mathbb{T}^b$  such that
\begin{equation}\label{gdec3}
	\mathrm{Leb}(X_N)\leqslant \mathrm{e}^{-{N}^{\sigma_2}},
\end{equation}
and for any $x\notin  X_N \bmod \Z^b$ and $Q_N\in \mathcal{E}_N^0$,
\begin{align*}
	\| G_{Q_N}(z) \|   & \leqslant \mathrm{e}^{N^{\sigma_1}},                                       \\
	|G_{Q_N}(z)(n,n')| & \leqslant \mathrm{e}^{-c_2 \|n-n'\|}, \text{ for } \|n-n'\|\geqslant N/10, 
\end{align*}
where  $z=E+\mathrm{i}\epsilon$ with $E\in[-K,K]$ (recall that $\sigma(H_x)\subseteq [-K+1,K-1]$) and $0<\epsilon\leqslant \epsilon_0$.

\begin{theorem}\label{thm1}
	Let $\mathcal{S}\subseteq [0,1]^{b}$ be a semi-algebraic set of degree $B$ and $\mathrm{Leb}(\mathcal{S}) < \eta$.
	Assume that  for any $N$ with
	\begin{equation*}
		\log B \lesssim \log N < \log \frac{1}{\eta},
	\end{equation*}
	we have the following sublinear bound:
	\begin{equation}\label{gsub}
		\#\{  n \in[-N, N]^{d}:  f^{n}x \in\mathcal{S}\} \leqslant C_{2} N^{1-\delta}.
	\end{equation}
	Assume LDT holds.
	Then  for any $\phi$ with compact support and any $\varepsilon>0$ there exists $T_0>0$ (depending on $d,b,p,\phi,K$, $\sigma_{1}, \sigma_{2},\delta, \epsilon_0$, $c_1$, $c_2$, $C_1$, $C_{2}$, $N_0$ and $\varepsilon$) such that for any $T\geqslant T_0$,
	\begin{align*}
		\langle|\tilde{X}_{H_{x}}|_{\phi}^p\rangle(T) & \leqslant  (\log T)^{  \frac{p}{\delta} +\varepsilon}, \\
		\langle|{X}_{H_{x}}|_{\phi}^p\rangle(T)       & \leqslant  (\log T)^{  \frac{p}{\delta} +\varepsilon}. 
	\end{align*}
\end{theorem}
\begin{proof}
	By approximating the analytic function with trigonometric polynomials, employing Taylor expansions, and applying standard perturbation arguments, we can assume that $X_N$ (as defined in \eqref{gdec3}) is a semi-algebraic set with degree less than $N^C$. Let $N_1=[ N^\varepsilon]$. Applying $\mathcal{S}=X_{N_1}$, $\eta=\mathrm{e}^{-N_1^{\sigma_2}}$ to \eqref{gsub}, one has \eqref{gsublinear} holds. Now Theorem \ref{thm1} follows from Theorem \ref{cor1}.
\end{proof}

\subsection{Proof of Theorem \ref{MainTheorem}}
Recall the skew-shift $f: \mathbb{T}^{b}\rightarrow \mathbb{T}^{b}$ is defined as
\begin{equation*}
	fx=f(x_{1},x_{2},\cdots, x_{b})=(x_{1}+\omega, x_{2}+x_{1},\cdots, x_{b}+x_{b-1}).
\end{equation*}
By the direct calculation, the $n$ step iteration of the skew-shift is
\begin{equation*}
	\begin{split}
		f^{n}x&=f^{n}(x_{1},x_{2},\cdots, x_{b})\\
		&=(x_{1}+C_{n}^{1}\omega, x_{2}+C_{n}^{1} x_{1}+C_{n}^{2}\omega, \cdots, x_{b}+C_{n}^{1}x_{b-1}+\cdots + C_{n}^{b}\omega).
	\end{split}
\end{equation*}
Obviously, for any $1\leqslant i\leqslant b$, the projection onto the $i$th coordinate for the $n$ step skew-shift is a polynomial in $n$ of degree $i$ whose highest degree term is $(\omega/i!)n^{i} $. In particular, the projection onto the $b$th coordinate is
\begin{equation*}
	(f^{n}x)_{b}=x_{b}+nx_{b-1}+ \frac{n(n-1)}{2}x_{b-2}+\cdots+ \frac{n(n-1)\cdots (n-b+1)}{b!} \omega.
\end{equation*}

Now Theorem \ref{MainTheorem} follows as a corollary of our semi-algebraic set estimates. 

\begin{proof}[Proof of Theorem \ref{MainTheorem}]
		
	By Lemma \ref{keepDC} and $\omega\in DC(\gamma, \tau)$, we know that for any $1\leqslant i\leqslant b$, $\omega/i!\in DC(\tilde{\gamma}, \tau)$ for some $\tilde{\gamma}=\tilde{\gamma}(b,\gamma)$.  
		
	For $2\leqslant b\leqslant 5$, we apply Theorem \ref{wm} (for $b\geqslant 6$, we apply Theorem \ref{vm}) with
	\begin{equation*}
		m=b, \ \mathbf{P}(n) \bmod \mathbb{Z}^{b}= f^{n}x,  \ \deg(P_{i})=i, \ \alpha_{i}=\omega/i!,
	\end{equation*}
	so that
	\begin{equation}\label{b1}
		\#\{ n \in[-N, N]:  f^{n}x \in\mathcal{S}\} \leqslant C(\varepsilon, b,\gamma,\tau) N^{1-\frac{1}{\tau b 2^{b-1}}+\varepsilon}.
	\end{equation}
	and
	\begin{equation}\label{b2}
		\#\{ n \in[-N, N]:  f^{n}x \in\mathcal{S}\} \leqslant C(\varepsilon, b,\gamma,\tau) N^{1-\frac{1}{\tau b^{2}(b-1)}+\varepsilon}.
	\end{equation}
	Recall
	\begin{equation}\label{psidef}
		\psi(b)=\begin{cases}
		2^{b-1},& \text{if} \ 2\leqslant b\leqslant 5,\\
		b(b-1),&\text{if} \ b\geqslant 6.
		\end{cases}
	\end{equation}
	Hence the theorem follows from Theorem \ref{thm1}, \eqref{b1}, \eqref{b2}, and \eqref{psidef}. 
\end{proof}

\appendix

\section{Proof of Lemma \ref{S2b}}\label{A1}
Denote $c_{n}=\mathrm{e}^{2\pi\mathrm{i} P(n;\alpha)}$.
For any $\rho\in \mathbb{Z}^{+}$, we can rewrite
\begin{equation}\label{sb}
	S^{\rho}=\sum_{\mathbf{n}} c_{n_{1}}c_{n_{2}} \cdots c_{n_{\rho}},
\end{equation}
where $\mathbf{n}$ runs over $\{1,2,\cdots,N\}^{\rho}$. For $\mathbf{n}=(n_{1},n_{2},\cdots,n_{\rho})$, we let $s_{j}(\mathbf{n})=n_{1}^{j}+n_{2}^{j}+\cdots+n_{\rho}^{j}$. Now we classify the set of $\mathbf{n}$ according to the value of $\mathbf{s}=(s_{1}(\mathbf{n}),s_{2}(\mathbf{n}),\cdots, s_{b-2}(\mathbf{n}))$. Let 
\begin{equation*}
	\mathfrak{S}=\{1,2,\cdots,\rho N\} \times \{1,2,\cdots, \rho N^{2}\}\times \cdots \times \{1,2,\cdots,\rho N^{b-2}\}.
\end{equation*}
Since $1\leqslant n_{i}^{j}\leqslant N^{j}$, we may restrict our attention to points $\mathbf{s}\in\mathfrak{S}$. For $\mathbf{s}\in\mathfrak{S}$, we stratify the set of $\mathbf{n}$ by letting
\begin{equation*}
	\mathcal{N}(\mathbf{s})=\{\mathbf{n}\in \{1,\cdots,N\}^{\rho}: s_{j}(\mathbf{n})=s_{j}, 1\leqslant j\leqslant b-2\}.
\end{equation*}
Then the sum in \eqref{sb} can be partitioned into subsums:
\begin{equation*}
	S^{\rho}=\sum_{\mathbf{s}\in\mathfrak{S}} \sum_{\mathbf{n}\in\mathcal{N}(\mathbf{s})} c_{n_{1}}c_{n_{2}} \cdots c_{n_{\rho}}.
\end{equation*}
Since
\begin{equation*}
	\#\mathfrak{S}=\rho N\times \rho N^{2}\times\cdots\times \rho N^{b-2}=\rho^{b-2} N^{(b-1)(b-2)/2}, 
\end{equation*}
by the H\"older inequality, we have
\begin{equation}\label{S2bdef}
	\begin{split}
		|S|^{2\rho}&\leqslant \rho^{b-2}N^{(b-1)(b-2)/2} \sum_{\mathbf{s}\in \mathfrak{S}} \Big|\sum_{\mathbf{n}\in \mathcal{N}(\mathbf{s})}c_{n_{1}}c_{n_{2}}\cdots c_{n_{\rho}}\Big|^{2}\\
		&=\rho^{b-2}N^{(b-1)(b-2)/2} \sum_{\substack{\mathbf{m},\mathbf{n}\\ s_{j}(\mathbf{m})=s_{j}(\mathbf{n})\\ 1\leqslant j\leqslant b-2}} c_{m_{1}}\cdots c_{m_{\rho}} \overline{c_{n_{1}}} \cdots \overline{c_{n_{\rho}}}.
	\end{split}
\end{equation}
To estimate \eqref{S2bdef}, we only need to estimate 
\begin{equation}\label{zdef}
	\mathcal{Z}:=\sum_{\substack{\mathbf{m},\mathbf{n}\\ s_{j}(\mathbf{m})=s_{j}(\mathbf{n})\\ 1\leqslant j\leqslant b-2}} c_{m_{1}}\cdots c_{m_{\rho}} \overline{c_{n_{1}}} \cdots \overline{c_{n_{\rho}}}.
\end{equation}
In the following, we estimate $\mathcal{Z}$. By the elimination from the symmetry of $\mathbf{m}$ and $\mathbf{n}$, we can see that
\begin{equation}\label{sym}
	\mathcal{Z}=\mathrm{e}^{2\pi\mathrm{i}(\alpha_{b}(s_{b}(\mathbf{m})-s_{b}(\mathbf{n}))+\alpha_{b-1}(s_{b-1}(\mathbf{m})-s_{b-1}(\mathbf{n})))}.
\end{equation}

Now we shift the every component of $\mathbf{m}$ and $\mathbf{n}$ by $m_{1}$. More precisely, we let $m_{i}=m_{1}+u_{i}$ and $n_{i}=m_{1}+v_{i}$ for $1\leqslant i\leqslant \rho$. Then by Binomial Theorem,
\begin{equation*}
	s_{j}(\mathbf{u})=\sum_{i=1}^{\rho}(m_{i}-m_{1})^{j}=\sum_{r=0}^{j}C_{j}^{r} s_{r}(\mathbf{m})(-m_{1})^{j-r},
\end{equation*}
and similarly,
\begin{equation*}
	s_{j}(\mathbf{v})=\sum_{i=1}^{\rho}(n_{i}-m_{1})^{j}=\sum_{r=0}^{j}C_{j}^{r} s_{r}(\mathbf{n})(-m_{1})^{j-r}.	
\end{equation*}
One can see that for $1\leqslant j\leqslant b$,
\begin{equation}\label{uvmn}
	s_{j}(\mathbf{u})-s_{j}(\mathbf{v})=\sum_{r=0}^{j}C_{j}^{r} \bigg(s_{r}(\mathbf{m})-s_{r}(\mathbf{n})\bigg)(-m_{1})^{j-r}.
\end{equation}
Hence the relations $s_{j}(\mathbf{m})=s_{j}(\mathbf{n})$ for $1\leqslant j\leqslant b-2$ imply that $s_{j}(\mathbf{u})=s_{j}(\mathbf{v})$ for $1\leqslant j\leqslant b-2$. In addition, we observe from \eqref{uvmn} that
\begin{equation*}
	s_{b-1}(\mathbf{u})-s_{b-1}(\mathbf{v}) =s_{b-1}(\mathbf{m})-s_{b-1}(\mathbf{n}),
\end{equation*}
and that
\begin{equation*}
	s_{b}(\mathbf{u})-s_{b}(\mathbf{v}) =s_{b}(\mathbf{m})-s_{b}(\mathbf{n})-bm_{1}(s_{d-1}(\mathbf{m})-s_{b-1}(\mathbf{n})).
\end{equation*}
For brevity we let $t_{j}=t_{j}(\mathbf{u},\mathbf{v})=s_{j}(\mathbf{u})-s_{j}(\mathbf{v})$. Then \eqref{sym} may be written as
\begin{equation}\label{Suv}
	\begin{split}	\mathcal{Z}&=\sum_{\substack{\mathbf{u},\mathbf{v}\\ t_{j}=0\\ 1\leqslant j\leqslant b-2}} \mathrm{e}^{2\pi\mathrm{i}(t_{b}\alpha_{b}+t_{b-1}\alpha_{b-1})} \sum_{m_{1}=1}^{N} \mathrm{e}^{2\pi\mathrm{i}(bt_{b-1} m_{1}\alpha_{b})}.
	\end{split}
\end{equation}
It is obvious that the geometric series above satisfies
\begin{equation}\label{Sinner}
	\bigg|\sum_{m_{1}=1}^{N} \mathrm{e}^{2\pi\mathrm{i}(bt_{b-1} m_{1}\alpha_{b})}\bigg| \lesssim \min \bigg(N, \frac{1}{\|bt_{b-1}\alpha_{b}\|_{\mathbb{T}}}\bigg).
\end{equation}
Substituting \eqref{Sinner} into \eqref{Suv} showing that
\begin{equation}\label{Z1}
	\mathcal{Z}\lesssim \sum_{\substack{\mathbf{u},\mathbf{v}\\ t_{j}=0\\ 1\leqslant j\leqslant b-2}} \min \bigg(N, \frac{1}{\|bt_{b-1}\alpha_{b}\|_{\mathbb{T}}}\bigg).
\end{equation}

Let $h\in \Z$ be a parameter with $|h|\leqslant 2\rho N^{b-1}$. Let $R_{1}(h)$ be the number of solutions of the system of equations 
\begin{equation*}
	\begin{split}	 &u_{2}^{j}+\cdots+u_{\rho}^{j}=v_{1}^{j}+\cdots+v_{\rho}^{j}, \ 1\leqslant j\leqslant b-2,\\
		&u_{2}^{b-1}+\cdots+u_{\rho}^{b-1}=h+v_{1}^{b-1}+\cdots+v_{\rho}^{b-1},
	\end{split}
\end{equation*}
in integer variables for which $|u_{i}|\leqslant N$ and $|v_{i}|\leqslant N$. Then by \eqref{Z1} and the definition of $R_{1}(h)$,
\begin{equation}\label{zr1}
	\mathcal{Z}\lesssim \sum_{|h|\leqslant 2\rho N^{b-1}}R_{1}(h) \min \bigg(N,\frac{1}{\|b h\alpha_{b}\|_{\mathbb{T}}}\bigg).
\end{equation}

Now we recover the variables $\mathbf{m},\mathbf{n}$ from $\mathbf{u},\mathbf{v}$. Treat $m_{1}$ as a parameter. We let $m_{i}=m_{1}+u_{i}$ and $n_{i}=m_{1}+v_{i}$ for $1\leqslant i\leqslant \rho
$. Then $R_{1}(h)$ is the number of solutions of the system
\begin{equation}\label{ns}
	\begin{split}
		&s_{j}(\mathbf{m})=s_{j}(\mathbf{n}), \ 1\leqslant j\leqslant b-2,\\
		&s_{b-1}(\mathbf{m})=h+s_{b-1}(\mathbf{n}),
	\end{split}
\end{equation}
in integer variables for which
\begin{equation*}
	\begin{split}
		&m_{1}-N\leqslant m_{i}\leqslant m_{1}+N, \ 1\leqslant i\leqslant \rho,\\
		&m_{1}-N\leqslant n_{i}\leqslant m_{1}+N, \ 1\leqslant i\leqslant \rho.\\
	\end{split}
\end{equation*}
Since $1\leqslant m_{i}, n_{i}\leqslant N$, we only need to consider $N+1\leqslant m_{1}\leqslant 2N$. Recall that Hardy–Ramanujan–Littlewood circle method relates the number of solutions to the integral on the circle. To apply the circle method, we let $R_{2}(h)$ be the number of solutions of \eqref{ns} subject to the weaker constraints:
\begin{equation*}
	\begin{split}
		&1\leqslant m_{i}\leqslant 3N,\ 1\leqslant i\leqslant \rho,\\
		&1\leqslant n_{i}\leqslant 3N,\ 1\leqslant i\leqslant \rho.\\
	\end{split}
\end{equation*}
Clearly, by the standard circle method, 
\begin{equation}\label{r2}
	R_{2}(h)=\int_{\mathbb{T}^{d-1}} \Big|\sum_{n=1}^{3N} \mathrm{e}^{2\pi\mathrm{i}\tilde{P}(n; \beta)} \Big|^{2\rho} \mathrm{e}^{-2\pi\mathrm{i} h\beta_{b-1}} \ \mathrm{d}\beta_{1}\cdots \mathrm{d}\beta_{b-1},
\end{equation}
where $\tilde{P}(n;\beta):=\sum_{j=1}^{b-1}\beta_{j} n^{j}$ is a polynomial of degree of $b-1$. Since for each $N+1\leqslant m_{1}\leqslant 2N$ we have $R_{2}(h)\geqslant R_{1}(h)$, thus $R_{2}(h)\geqslant N R_{1}(h)$. It is evident that  by taking the absolute value in \eqref{r2}, 
\begin{equation}\label{r20}
	R_{2}(h)\leqslant R_{2}(0), \ \text{for all} \ h.
\end{equation}
But by \eqref{Jdef1} and \eqref{Jdef2} we know $R_{2}(0)=J_{b-1}(3N;\rho)$. So by \eqref{zr1}, \eqref{r2}, and \eqref{r20}, we have
\begin{equation}\label{zJ}
	\mathcal{Z}\lesssim \frac{1}{N} J_{b-1}(3N;\rho) \sum_{|h|\leqslant 2\rho N^{b-1}} \min \bigg(N, \frac{1}{\|b h\alpha_{b}\|_{\mathbb{T}}}\bigg).
\end{equation}

Finally, we substitute \eqref{zJ} and \eqref{zdef} into \eqref{S2bdef}, one can get
\begin{equation*}
	|S|^{2\rho}\lesssim N^{(b-1)(b-2)/2} \frac{1}{N} J_{b-1}(3N;\rho)\sum_{|h|\leqslant 2\rho b N^{b-1}} \min \bigg(N, \frac{1}{\| h\alpha_{b}\|_{\mathbb{T}}}\bigg).
\end{equation*}

\section*{Acknowledgments}
W. Liu was a 2024-2025 Simons fellow.
This work was supported in part by NSF DMS-2246031,  DMS-2052572 and  DMS-2000345.

\section*{Statements and Declarations}
{\bf Conflict of Interest} 
The authors declare no conflicts of interest.

\vspace{0.2in}
{\bf Data Availability}
Data sharing is not applicable to this article as no new data were created or analyzed in this study.

\bibliographystyle{alpha} 
\bibliography{main}

\end{document}